\documentclass[authoryear,preprint]{elsarticle}
\usepackage{amssymb,amsmath} 
\usepackage[breaklinks]{hyperref} 
\usepackage{tensor} 
\usepackage{natbib}
\setcitestyle{aysep={}} 

\newcommand{\HH}{\mathcal{H}}
\newcommand{\PP}{\mathcal{P}}
\newcommand{\UU}{\mathcal{U}}
\newcommand{\RR}{\mathbb{R}}
\newcommand{\Jab}{\tensor{J}{^a_b}}
\newcommand{\Inn}[2]{\langle #1, #2 \rangle}

\newtheorem{prop}{Proposition}
\newtheorem{lemm}{Lemma}
\newtheorem{coro}{Corollary}
\newdefinition{defi}{Definition}
\newproof{proof}{Proof}

\begin{document}
\begin{frontmatter}
\title{A general perspective on time observables}
\author[bwr]{Bryan W. Roberts}
\ead{b.w.roberts@lse.ac.uk}
\address[bwr]{Department of Philosophy, Logic and Scientific Method, London School of Economics and Political Science, Houghton Street, London, WC2A 2AE}
\begin{abstract}
I propose a general geometric framework in which to discuss the existence of time observables. This frameworks allows one to describe a local sense in which time observables always exist, and a global sense in which they can sometimes exist subject to a restriction on the vector fields that they generate. Pauli's prohibition on quantum time observables is derived as a corollary to this result. I will then discuss how time observables can be regained in modest extensions of quantum theory beyond its standard formulation. 
\end{abstract}
\end{frontmatter}	

\section{Introduction}\label{sec:introduction}

The characterization of time as a measurable quantity or ``observable'' has been the subject of much discussion in the foundations of physics. In a famous remark in his textbook on quantum theory, \citet[p.63]{pauli-qm} argued that quantum theory lacks the capacity to describe such an observable. Subsequent work clarified the precise mathematical results underlying this remark. In quantum theory, there is a certain plausible property of clocks and timers, which I will define precisely below and refer to as ``timeliness,'' which is incompatible with the assumption that the Hamiltonian is half-bounded \citep{srinivas1981time}. Since most known quantum systems are described by a Hamiltonian that is bounded from below, it follows that the standard description a quantum system prohibits this property of clocks. This prohibition on quantum clocks is sometimes referred to as \emph{Pauli's theorem}.

I would like to propose a more general and I think positive perspective on this result. My approach will be to frame the question in the general geometric language of dynamical systems. This language will be familiar as one in which classical Hamiltonian mechanics is often framed. However, we will see that with a little extra structure it is enough to capture quantum theory too, as well as some extensions of quantum theory beyond its orthodox formulation. My aim will be to start determining the lay of the land for time observables in this general framework. I will first illustrate a local sense in which time observables generically exist, and discuss the extent to which they are unique. Next I will show that time observables can also exist globally, subject to a constraint on the vector fields they generate. Pauli's theorem can then be viewed as as a corollary to this result; to make it precise, I will adopt the geometric language for quantum mechanics developed by \citet{AshtekSchill1999a}.

My central conclusion will be that Pauli's prohibition on time observables is not as severe as it may seem. Viewed from a more general perspective, time observables can be regained not only in classical mechanics, but also in many modest extensions of quantum theory that go just a bit beyond its usual formulation.

\section{Timely observables in general dynamical systems}\label{sec:time-observables-and-existence}

The dynamics of a physical system will be described by a triple $(\PP,\Omega_{ab},h)$, where $\PP$ is a smooth manifold, $\Omega_{ab}$ is a symplectic form, and $h:\PP\rightarrow\RR$ is a smooth function (the Hamiltonian) that generates the dynamics. I will refer to such a triple as a \emph{dynamical system}.

One represents time in a dynamical system as follows. The Hamiltonian $h$ generates a Hamiltonian vector field, which can be written, $H^a=\Omega^{ba}d_bh$. The integral curves $c_t:\RR\rightarrow\PP$ that thread the vector field are defined along a parameter $t$, which is interpreted as time. In dynamical systems, \emph{time} is just what parametrizes the solutions to Hamilton's equations.

Measurable quantities or ``observables'' in a dynamical system (like position and angular momentum) are represented by smooth functions $f:\PP\rightarrow\RR$. One can think of such a function as describing the value $f(x)$ displayed by a measuring device when the system is in the state $x$. This motivates taking a ``time observable'' to be a smooth function $\tau:\PP\rightarrow\RR$ that describes a particular measuring device, like a timer or clock, which measures the passage of time. Although there are many ways to characterize such a function, our main concern will be the following property.

\begin{defi}[timeliness]\label{defi:global-timeliness}
	Let $(\PP,\Omega_{ab},h)$ be a dynamical system. To say that a smooth function $\tau:\PP\rightarrow\RR$ is \emph{timely} means that, if $c_t$ is an integral curve of the Hamiltonian vector field generated by the Hamiltonian $h$, then $\tau(c_t)=\tau(c_0)+t$ for all $t\in\RR$.
\end{defi}

A timely observable $\tau$ behaves like a timer, in the following sense. Suppose a timer initially displays `$0900$.' After a duration of time $t$ passes, the timer adds the quantity $t$ to its display, showing $0900+t$. A globally timely observable captures this by changing its value from $\tau(c_0)$ to $\tau(c_t)=\tau(c_0)+t$.

Notice that in order to satisfy Definition \ref{defi:global-timeliness}, a necessary (but not sufficient) requirement is that $\tau(c_t)$ be a monotonically increasing function of $t$. As a result, a timely observable cannot exist at a stationary point, where the Hamiltonian vector field vanishes and the integral curve $c_t=c_0$ is constant for all times $t$. It also cannot exist in a phase space of finite Lebesgue volume, due to the famous no-go results of Poincar\'e and Zermelo, which under such circumstances prohibit the existence of any function that increases monotonically with time\footnote{I thank Nikola Buri\'{c} and Jos Uffink for pointing this out to me in discussion.}. However, a timely observable can easily be constructed ``locally,'' which is to say for finite amounts of time, in the neighborhood of every non-stationary point. Here is one such construction\footnote{This rough idea was independently suggested by \citet[\S 1]{ArsenovicBuricDavidovicPranovic-2012timeobservables}. It is an immediate consequence of the Darboux theorem.}.

\begin{prop}\label{prop:1}
	Let $(\PP,\Omega_{ab},h)$ be a dynamical system. If the integral curve generated by $h$ through a point $x\in\PP$ is not constant at that point, then there exists a function $\tau:B\rightarrow\RR$  in some neighborhood $B$ of $x$ such that $\tau(c_t) = \tau(c_{t^\prime}) + (t-t^\prime)$ for all $c_t, c_{t^\prime} \in B$. In particular, the timeliness condition $\tau(c_t) = \tau(c_{0}) + t$ holds when $c_0,c_t \in B$.
\end{prop}
\begin{proof}
Since $x$ is not stationary, there is a $(2n-1)$-dimensional hypersurface $\mathcal{S}$ through $x$ that does not contain a tangent vector to the integral curve at $x$. The existence and uniqueness of solutions to ordinary differential equations imply that in some neighborhood $B$ of $x$, each point $y \in B$ is intersected by a unique integral curve $c_t$. So, define a function $\tau:B\rightarrow\RR$ as follows: for each $y\in B$, if $c_t$ is the integral curve passing through $y$ at time $t$, then let $\tau(y)$ be the amount of time required to go from the hypersurface $\mathcal{S}$ to the point $y$ along that curve. Then $\tau$ is our desired function, since,
	\[ \tau(c_t) = \tau(c_{t^\prime}) + (\underbrace{\tau(c_t) - \tau(c_{t^\prime})}_{t-t^\prime}) = \tau(c_{t^\prime}) + (t-t^\prime),\]
for all $c_t,c_{t^\prime} \in B$. \hfill $\Box$
\end{proof}

This construction guarantees the existence of a locally operating clock, which keeps time for events in the neighborhood of a non-stationary point. Of course, $\tau$ need not satisfy timeliness for all times $t\in\RR$. The integral curves $c_t$ may soon leave the region, and it is not generally possible to extend a timely observable to all of phase space. Happily, a locally operating clock is all that we are typically dealing with when we measure time anyway.

Like many local facts, Proposition \ref{prop:1} can be expressed in terms of derivatives and the Poisson bracket. In particular, if $c_t$ is an integral curve of the vector field generated by $h$, and if $\tau$ satisfies $\tau(c_t)=\tau(c_{t^\prime})+(t-t^\prime)$ in some region $B\subseteq\PP$, then applying the definition of the Poisson bracket yields,
	\[ \{h,\tau\} = H^ad_a\tau = \frac{d}{dt}\tau\left(c_t\right)= \frac{d}{dt}(\tau(c_{t^\prime})+t-t^\prime) = 1. \]
Thus, establishing that an observable $\tau$ is timely in a region also establishes that $\{h,\tau\}=1$ in that region. I will call this property \emph{local timeliness}.

If a timely function exists, it is in general not unique. For example, if a dynamical system has a Hamiltonian $h$ and a timely function $\tau$, then $\tilde{\tau}=\tau+h$ is also timely. However, timely functions are still unique up to a natural notion of equivalence, which is characterized by the following.

\begin{prop}\label{prop:uniqueness}
	Let $\tau$ be a timely function. Then $\tilde{\tau}$ is also timely if and only if $\tilde{\tau} = \tau+f$ for some function $f$ that is constant in time, i.e., $f(c_t)=f(c_0)$ for all $t\in\RR$.
\end{prop}
\begin{proof}
The `if' direction follows immediately from the definition,
	\[ \tilde{\tau}(c_t) = \tau(c_t)+f(c_t) = \tau(c_0) + t + f(c_0) = \tilde{\tau}(c_0) + t.\]
The `only if' direction follows by subtracting $\tau(c_t)=\tau(c_0)+t$ from $\tilde{\tau}(c_t) = \tilde{\tau}(c_0) + t$ to get $\tilde{\tau}(c_t) - \tau(c_t) = \tilde{\tau}(c_0) - \tau(c_0)$. Then set $f(c_t) = \tau_2(c_0) - \tau_1(c_0)$ for all $t$. \hfill $\Box$
\end{proof}
This means that, on each dynamical trajectory $c_t$, the values assigned by two timely functions $\tau$ and $\tilde{\tau}$ differ only by a constant $f(c_0)$. Their time derivatives are the same, $d\tau/dt = d\tilde{\tau}/dt$. Thus, on a given trajectory, two timely functions are unique up to a choice of initial time.

\section{Timely observables generate incomplete vector fields}\label{sec:time-observables-generate-incomplete-fields}

Recall that a vector field $F^a$ is \emph{complete} if every maximal integral curve $c_s$ of $F^a$ has the entire real number line as its domain; otherwise, $F^a$ is \emph{incomplete}. The purpose of this section is to illustrate that generating an incomplete vector field is a generic property of timely (and locally timely) observables, whenever the Hamiltonian is bounded from below. The proof is just a geometric analogue of a well-known proof of Pauli's theorem in quantum theory \citep[see, e.g.,][\S 4.3]{butterfield2013time}\footnote{I think David Malament for comments that helped to correctly state this proposition.}.

\begin{prop}\label{prop:2}
Let $(\PP,\Omega_{ab},h)$ be a dynamical system, and let $h$ be bounded from below. If $\tau:\PP\rightarrow\RR$ is a smooth function that satisfies $\{h,\tau\}=1$,
then the Hamiltonian vector field generated by $\tau$ is incomplete.
\end{prop}
\begin{proof}
Let $\{h,\tau\} = 1$. Then, by the skew-symmetry of the bracket, $\{\tau,h\} = -1$. Let $c_s$ be an integral curve of the Hamiltonian vector field generated by $\tau$, and define $h(s):=h(c_s)$. Then $-1 = \{\tau,h\} = \frac{d}{ds}h(s)$, and so,
\begin{equation*}
h(s) - h(0) = \int_0^s\frac{d}{ds}h(s)ds = -\int_0^s ds = -s.
\end{equation*}
If $\tau$ were to generate a complete Hamiltonian vector field, then $s$ could take any real value, and we would have $h(s) = h(0) - s$ for all $s\in\RR$. This would imply that the range of $h$ is the entire real line, contradicting the assumption that $h$ is half-bounded. Therefore, the vector field generated by $\tau$ is incomplete. \hfill $\Box$
\end{proof}

Since timeliness implies local timeliness, an analogous local statement follows as an immediate corollary.
\begin{coro}
	Let $(\PP,\Omega_{ab},h)$ be a dynamical system, and let $h$ be bounded from below. If  $\tau:\PP\rightarrow\RR$ is a smooth function such that for all integral curves $c_t$ of $h$, $\tau$ satisfies $\tau(c_t) = \tau(c_0) + t$ for all $t\in\RR$, then the Hamiltonian vector field generated by $\tau$ is incomplete. 
\end{coro}

Notably, this fact does not prohibit time observables outright when the Hamiltonian is half-bounded. Here is an easy example of a smooth timely function on a manifold with a half-bounded Hamiltonian. Let $h(q,p)=(1/2m)p^2$ be the free particle Hamiltonian, and let us restrict the phase space $\PP$ to the positive-momentum half of the real plane,
	\[ \PP = \RR\times\RR^+ = \{(q,p) : q\in(-\infty,\infty), p\in(0,\infty)\}. \]
Then $\tau(q,p) = mq/p$ is a smooth function that is timely. To check this, apply the facts that $p_t = p_0$ (a constant) and $q_t = p_0t/m + q_0$ to the definition of $\tau$, to get:
	\[ \tau(q_t,p_t) = mq_t/p_t = m(p_0t/m + q_0)/p_0 = \tau(q_0,p_0)+t.\]
So, $\tau$ is timely. By force of Proposition \ref{prop:2}, the Hamiltonian vector field generated by $\tau$ cannot be complete; one can also check explicitly that the integral curves of this vector field are defined only for a half-bounded set of parameter values.

It's remarkable that this example is not available in quantum theory. The corresponding symmetric Hilbert space operator $T=(m/2)(QP^{-1}+P^{-1}Q)$ is not a quantum observable in the conventional sense, in that it is not self-adjoint and admits no self-adjoint extension\footnote{For a more general discussion of this operator, see \citep{BuschEtAl1994}.}. This is a consequence of the fact that the momentum operator $P$ is not self-adjoint when restricted to the positive-momentum half of its spectrum \citep[Ex. 4.2.5]{BlankExnerHavlicek}. These failures can be attributed to Pauli's prohibition on the existence of a timely self-adjoint operator in quantum theory, whenever the Hamiltonian is half-bounded. In the next section, we will see that this prohibition is in fact just a special case of Proposition \ref{prop:2}.

\section{No time observables in quantum theory}\label{sec:no-time-observables}

The quantum analogue of a timely observable is a densely defined self-adjoint operator $T$ satisfying,
	\[ e^{-itH}Te^{itH}\psi=(T+tI)\psi, \]
for all $t\in\RR$ and for all $\psi$ in the domain of $T$. The uniqueness result expressed by Proposition \ref{prop:uniqueness} carries over here as well: if $T$ is timely, then $\tilde{T}$ is also timely if and only if $\tilde{T} = T + F$ for some $F=F(t)$ that is constant in time\footnote{The proof of this exactly follows the steps of Proposition \ref{prop:uniqueness}.}. But in this case, more can be said: by the Stone-von Neumann theorem \citep[Theorem 8.2.4]{BlankExnerHavlicek}, two timely observables $T$ and $\tilde{T}$ are unitarily equivalent, since they both satisfy the canonical commutation relations in Weyl form.

However, Pauli's theorem prohibits the existence of such timely observables in quantum systems with a half-bounded Hamiltonian. This section will develop a general perspective on why that is. In rough sketch, we will observe a sense in which quantum theory is very restrictive as to what it counts as an observable than more general dynamical systems. A quantum observable is conventionally represented by a self-adjoint operator, which always generates the analogue of a complete vector field. Namely, a self-adjoint $A$ always generates a unitary group $U(s)=e^{isA}$ that is defined for all parameter values $s\in\RR$. This fact, together with Proposition \ref{prop:2}, implies the conclusion of Pauli's theorem: when the Hamiltonian is half-bounded, there can be no self-adjoint operator that is timely.

To make this argument precise, we adopt a formalism proposed by \citet{AshtekSchill1999a}, in which quantum mechanics is viewed quite literally a particular class of dynamical system called a K\"ahler manifold. We begin with a few remarks about this formalism to fix notation. For details, the reader is referred to \citet{schilling-dissertation,AshtekSchill1999a}.

\subsection{K\"ahler quantum mechanics}

Suppose we view an ordinary quantum system as a pair $(\HH,e^{-itH})$, where $\HH$ is a separable Hilbert space and $\UU_t=e^{-itH}$ is a strongly continuous unitary group that provides the dynamics. Following \citet{AshtekSchill1999a}, one can view this structure as a particular kind of dynamical system $(\PP,\Omega_{ab},h)$.

We begin by defining $\PP$ to be the set of projective rays of $\HH$, which forms a Hilbert manifold whenever $\HH$ is separable \citep{schilling-dissertation}. But instead of viewing $\PP$ as a complex manifold, we view it as a real manifold equipped with a complex structure $\Jab$, which behaves like multiplication by $i$. Next, we write the inner product in terms of real functions $g$ and $\Omega$ that yield the real and imaginary parts, respectively:
	\[ \Inn{\psi}{\phi} = g(\psi,\phi) + i\Omega(\psi,\phi). \]
It turns out that the function $\Omega$ is a symplectic form, $g$ is a strongly non-degenerate Riemannian metric, and these objects satisfy the relations,
\begin{align*}
	g_{ab} = \Omega_{ac}\tensor{J}{^c}{_b}, && g_{cd}\tensor{J}{^c_a}\tensor{J}{^d_b} = g_{ab}, && \Omega_{ab}\tensor{J}{^a_c}\tensor{J}{^b_d} = \Omega_{cd}.
\end{align*}
Such a structure $(\PP,\Omega_{ab},\Jab,g_{ab})$ is called a \emph{K\"ahler manifold}\footnote{See \citet[\S 5.3]{MarsdenRatiu1999} for an introduction to K\"ahler manifolds, and \citet{Ballma2006a} for a book-length treatment.}. The following lemma summarizes two important properties of a K\"ahler manifold constructed from a projective Hilbert space, which we will make use of shortly.

\begin{lemm}\label{lem:manifold}
	Let $(\PP,\Omega_{ab},\Jab,g_{ab})$ be a K\"ahler manifold, in which $\PP$ is a real Hilbert manifold corresponding to a separable Hilbert space, $g_{ab}$ is the Riemannian metric defined by the real part of the Hilbert space inner product, and $\Omega_{ab}$ is the symplectic form defined by the imaginary part of the inner product. Then,
	\begin{enumerate}
		\item[(a)] $\PP$ is a topologically connected, and
		\item[(b)] $\PP$ is geodesically complete with respect to $g_{ab}$.
	\end{enumerate}
\end{lemm}
\begin{proof}
	\citep[Theorem II.2]{AshtekSchill1999a}
\end{proof}

To have a dynamical system in the sense of the previous sections, we must now know the nature of observables. The natural choice is to define a function $f:\PP\rightarrow\RR$ for each self-adjoint operator $F$ on the Hilbert space, given by the amplitude,
	\[ f(\psi) = \Inn{\psi}{F\psi}\]
defined on unit vectors $\psi$. On the one hand, this function captures the probabilistic properties of the self-adjoint operator $F$ that make it an observable. But there is more: the function $f$ generates a Hamiltonian vector field on $\PP$ by the usual prescription $F^a:=\Omega^{ba}d_bf$. This vector field turns out to correspond precisely to the unitary group $U(s)=e^{isF}$ generated by $F$ on the underlying Hilbert space \citep[\S II.A.1]{AshtekSchill1999a}. In particular, the quantum Hamiltonian $H$ generates dynamical trajectories that correspond to precisely the Hamiltonian vector field generated by $h$. In this sense, a quantum system $(\HH, e^{-itH})$ really can be viewed as an example of a dynamical system $(\PP,\Omega_{ab},h)$.

We will make use of one further property of quantum observables, when viewed as functions on a manifold $\PP$.

\begin{lemm}\label{lem:killing}
	Let $(\PP,\Omega_{ab},\Jab,g_{ab})$ be a be the K\"ahler manifold constructed from a separable Hilbert space $\HH$, as in Lemma \ref{lem:manifold}. Then the statement,
	\[ f(\psi) = \Inn{\psi}{F\psi} \]
for some densely defined self-adjoint operator $F$ holds if and only if the vector field $F^a :=\Omega^{ba}d_bf$ is a Killing field, $\nabla_a F_b + \nabla_b F_a = \mathbf{0}$, where $\nabla_a$ is the (unique) derivative operator associated with $g_{ab}$.
\end{lemm}
\begin{proof}
\citep[Corollary 1]{AshtekSchill1999a}
\end{proof}
It will be helpful in what follows to sketch what underlies this fact; for a complete proof, see \citep{schilling-dissertation}. The metric and the symplectic forms are defined by the real and imaginary parts of the inner product, respectively. Thus, if a vector field is both Hamiltonian and Killing, then it preserves both parts of the inner product, and is thus implemented by a set of unitary transformations. Because of the Killing condition, these transformations form a one-parameter group\footnote{This fact, as well as a detailed discussion of Killing vector fields is given by \citet[Appendix C.2-3]{WaldR:1984}.}, and are therefore generated by a self-adjoint operator according to Stone's theorem. The converse may be argued similarly. Therefore, what identifies \emph{quantum} observables (corresponding to self-adjoint operators) from among the many functions on phase space is that they generate a Hamiltonian vector field that is Killing.

Let me summarize the observations of this section that are most essential for the next. An ordinary quantum system $(\HH, e^{-itH})$ can be identified with a dynamical system $(\PP,\Omega_{ab},h)$ equipped with an additional Riemannian metric $g_{ab}$ and a complex structure $\Jab$, which together form a K\"ahler manifold. The manifold $\PP$ is connected and geodesically complete with respect to the metric $g_{ab}$. Each self-adjoint operator $F$ on the original Hilbert space $\HH$ corresponds to a smooth function $f:\PP\rightarrow\RR$ on the manifold $\PP$, which generates a Hamiltonian vector field that is Killing.

\subsection{No quantum time observables}

We have already observed that a self-adjoint operator $F$ always generates a unitary group $U(s)=e^{isF}$ that is ``complete'' in that it is defined for all parameter values $s\in\RR$. Viewing quantum theory as a special class of dynamical systems, this corresponds to the following elementary fact.

\begin{lemm}\label{lemm:complete}
	If $F^a$ is a Killing field on a connected and geodesically complete Riemannian manifold, then $F^a$ is complete.
\end{lemm}
\begin{proof}
Since $F^a$ is a Killing field, $F^aF_a=k^2$ is constant along its integral curves. (The integral curves of a Killing field consist of isometries $\varphi$, which preserve inner products: $g_{ab}(\varphi_*\xi^a)(\varphi_*\rho^b)=(\varphi^*g_{ab})\xi^a\rho^b=g_{ab}\xi^a\rho^b$. Thus in particular $\varphi$ preserves the inner product of $F^a$ with itself.) Choose any point $p\in\PP$. If $F^a=\mathbf{0}$ at $p$, then its integral curve is the trivial curve that just sits at $p$, and so the domain of that curve is $\RR$. If $F^a\neq\mathbf{0}$, then there is a segment $\varphi:(r,s)\rightarrow\PP$ of the integral curve through $p$ of finite length,
\begin{equation*}
	\int_{r}^{s}(F^aF_a)^{1/2}dt = k(s-r)>0.
\end{equation*}
This implies that if $r_1$ is the midpoint of $(r,s)$, $r_2$ is the midpoint of $(r_1,s)$, and so on, then the sequence $\varphi(r_1),\varphi(r_2),\varphi(r_3),\dots$ is Cauchy. But the Hopf-Rinow theorem guarantees that every Cauchy sequence on a geodesically complete Riemannian manifold converges \citep[Theorem 6.13]{Lee-IntroToRiemannGeom}. So, our sequence converges to $\varphi(s)$. Therefore, the domain of the finite segment $\varphi:(r,s)\rightarrow\PP$ extends to $[r,s]$, and thus by the smoothness of $\varphi$ it extends to some open neighborhood of $[r,s]$. Continuing in this way, it follows that the domain of $\varphi$ is $\RR$. \hfill $\Box$
\end{proof}

We may now combine the lemmas to produce a geometric expression of Pauli's theorem in quantum theory. Lemma \ref{lem:killing} showed that the observables of quantum theory correspond to smooth functions on a dynamical system $(\PP,\Omega_{ab},h)$ that generate Killing fields. The remaining lemmas immediately imply that these quantum observables cannot be timely, when the Hamiltonian $h$ is half-bounded.

\begin{prop}[Pauli's Theorem]\label{prop:3}
	Let $(\PP,\Omega_{ab},\Jab,g_{ab})$ be a K\"ahler manifold as in Lemma \ref{lem:manifold}. If $h$ is a smooth half-bounded function, and if $\tau$ is a smooth function generating a Hamiltonian vector field that is Killing, then $\{h,\tau\}\neq1$.
\end{prop}
\begin{proof}
By Lemma \ref{lem:manifold}, $\PP$ is connected and geodesically complete with respect to $g_{ab}$. So, if the Hamiltonian vector field $T^a=\Omega^{ba}d_b\tau$ is Killing, then it is also complete by Lemma \ref{lemm:complete}. Proposition \ref{prop:2} thus guarantees that no such function $\tau$ exists satisfying $\{h,\tau\}=1$. \hfill $\Box$
\end{proof}
Since timeliness implies local timeliness, it follows that no $\tau$ can satisfy $\tau(c_t)=\tau(c_0)+t$ when $h$ is half-bounded either. Thus, restricting a dynamical system to the observables of orthodox quantum theory precludes the possibility of a time observable.

\section{Time observables in extensions of quantum theory}\label{sec:extensions-of-quantum-theory}

A half-bounded Hamiltonian precludes timely observables among the impoverished class of functions characterizing quantum observables. But we have seen that timely functions can still appear in more general dynamical systems. This suggests that the lack of time observables is a somewhat fragile aspect of quantum theory, which can be regained in modest extensions of the theory. Let me briefly comment on two such extensions, and provide examples to illustrate how they can eschew Pauli's theorem.

\subsection{Maximal symmetric operators} Suppose we retain the convention that a quantum observable is a linear operator $A$ that is \emph{symmetric}, in that $AA^*\psi=A^*A\psi$ for all $\psi$ in the domain of $A$. This assures that $A$ has a real spectrum. But suppose we give up the additional convention that $A$ be \emph{self-adjoint}, which requires also that $A$ and $A^*$ have the same domain\footnote{In finite-dimensional Hilbert spaces, an operator is symmetric if and only if it is self-adjoint (in which case it is often called \emph{Hermitian}). But this equivalence fails for infinite-dimensional Hilbert spaces.}. A symmetric operator that does not extend to any self-adjoint operator is called \emph{maximal symmetric}.

Stone's theorem assures us that the set of symmetries $U_s=e^{isA}$ generated by a self-adjoint operator $A$ is defined for every real number $s\in\RR$ \cite[Theorem 5.9.1]{BlankExnerHavlicek}. In a general dynamical system, this is analogous to the set of integral curves $c_s$ being defined for all $s\in\RR$, and thus by a complete vector field. However, \citet{cooper1947,cooper1948symmetric} showed that maximal symmetric operators lack this assurance\footnote{Thanks to Thomas Pashby for pointing this out to me.}, instead allowing for the analogue of an incomplete vector field. In particular, call an operator $U$ an \emph{isometry} if it is unitary on a closed subspace $V\subseteq\HH$; thus, an isometry does not necessarily have an inverse with a dense domain. Then we have:

\begin{itemize}
	\item \emph{Stone's Theorem for Maximal Symmetric Operators.} A maximal symmetric operator $A$ generates a set of operators $U_s$ that forms a strongly continuous set of isometries for all $s\geq0$ or for all $s\leq0$, but not for both\footnote{In particular, a maximal symmetric operator is associated with a positive operator-valued measure $\Delta\mapsto E_\Delta$ on Borel subsets of $\RR$, which allows one to define $U_s := \int_\RR e^{i\lambda s}dE_\lambda$ for all $s\in\RR$. \citet{cooper1948symmetric} showed that this set forms a strongly continuous semigroup (for which inverses are not assured) of isometries for either the positive or the negative values of $s$.}.
\end{itemize}

This feature of symmetric operators lifts the requirement of a complete vector field, thereby opening the door for timely observables. One can see just where the breakdown occurs in terms of K\"ahler quantum theory. Let $\tau(\psi)=\Inn{\psi}{T\psi}$ be the smooth function corresponding to a symmetric operator $T$. The Hamiltonian flow it generates corresponds to a set of operators that preserve the Hilbert space inner product, and therefore preserve the metric. But since these isometries do not form a group defined for all real parameter values, the corresponding Hamiltonian vector field is not Killing\footnote{See \citet[Appendix C.2-3]{WaldR:1984}.}. Thus, Lemma \ref{lem:killing} does not apply, and Pauli's theorem does not go through.

An example is provided by the positive-momentum free particle discussed in Section \ref{sec:time-observables-generate-incomplete-fields}. The operator $T=(m/2)(QP^{-1}+P^{-1}Q)$ is permitted to be timely because it is maximal symmetric. The vector field it generates is incomplete, and so Pauli's prohibition on time observables is escaped. Similar examples have been studied by \citet{BuschEtAl1994}. We now have a more general perspective on why such time observables are possible.

\subsection{Weinberg functions} A second, more general route to time observables begins with the geometric perspective. Suppose we are a even more inclusive with our observables, by dropping not only the requirement that an observable function generate a Killing field (as with self-adjoint operators), but also that it even generate a set of isometries (as with symmetric operators).  Suppose we allow any smooth function that generates a vector fields that covers the entire phase space. This class of functions has been studied by \citet[\S III.A]{AshtekSchill1999a}, who show that they characterize a class of extensions of quantum theory proposed by \citet{Weinberg1989a}. For this reason, they refer to these functions as \emph{Weinberg functions}.

Unlike orthodox quantum observables, the Weinberg functions can quite easily be timely. To illustrate, consider the following example, which is due to John D. Norton\footnote{Private communication.}. Let $\PP=\RR^2$, with a Cartesian coordinate system $(q,p)$ and the standard symplectic form, together with a (half-bounded) Hamiltonian $h(q,p)=e^{p}$. The integral curves generated by $h$ can be written $(q_t,p_t)=(e^{p_0} t + q_0,p_0)$ for an arbitrary initial point $(q_0,p_0)$. In this system, the smooth function $\tau(q,p)=q/e^p$ is timely:
\begin{equation*}
	\tau(q_t,p_t) = q_t/e^{p_t} = (e^{p_0} t + q_0)/e^{p_0} = q_0/e^{p_0} + t = \tau(q_0,p_0)+t.
\end{equation*}
The Hamiltonian vector field generated by the timely function $\tau$ has integral curves given\footnote{Check: $dq/ds = \partial \tau/\partial p = -q/e^p$ and $dp/ds = -\partial \tau/\partial q = -1/e^p$. One can easily see by differentiation that these equations are satisfied by $q_s=q_0(1-s/e^{p_0})$ and $p_s=\log(e^{p_0}-s)$.} by $q_s=q_0(1-s/e^{p_0})$ and $p_s=\log(e^{p_0}-s)$. Thus, the vector field tangent to these curves is incomplete, because the curve with the initial point $(q_0,p_0)=(0,0)$ cannot be extended beyond $s=1$ where $p(s)=\log(1-s)$ becomes undefined. However, it is smooth and defined on the entire manifold, and therefore counts as a Weinberg function according to the definition of \citeauthor{AshtekSchill1999a}.

\section{Conclusion}

The time observable question can be made precise a property we have called ``timeliness,'' which describes one simple feature of the behavior of clocks and timers. Once we have specified what it means to be a timely observable, the language of dynamical systems provides a general perspective from which to discuss their existence. We began by observing a local sense in which timely observables are guaranteed to exist in all dynamical systems. We then showed that a global timely observable can only exist if it generates an incomplete Hamiltonian vector field. This property is not possible among the observables among the observables of quantum theory, when the Hamiltonian is half-bounded. We made that fact precise in the geometric formulation of quantum theory, in which a quantum system is viewed as a particular kind of dynamical system that is in addition a K\"ahler manifold. The result is a novel proof of Pauli's theorem, visible now as a corollary to a more general fact about dynamical systems.

From this perspective, the quantum prohibition on time observables does not appear to be such a permanent feature of the description of the physical world. Time observables can be regained in many more general classes of observables on K\"ahler manifolds, which still retain other important aspects of quantum theory. So, insofar as these more general dynamical systems provide a plausible framework for exploring new physics, the prospects for regaining time observables as we go beyond quantum theory may be more promising than the Pauli's theorem by itself suggests.

\section*{Acknowledgements}

Thanks to Nikola Buri\'{c}, Jeremy Butterfield, David Malament, John D. Norton, Nicholas Teh, and Jos Uffink for comments that led to improvements, and especially to Thomas Pashby for many invaluable discussions.



\begingroup\raggedright\endgroup 
\end{document}